\documentclass[a4paper,12pt]{article}
\usepackage{latexsym,amssymb,amsfonts,amsmath,amsthm}
\usepackage[dvips]{graphicx}

\usepackage{amsmath} 

\usepackage{bm}

\newtheorem{theorem}{Theorem}[section]

\newtheorem{proposition}[theorem]{Proposition}
\title{A walk on max-plus algebra}
\author{ 
{\small 
Sennosuke Watanabe,$^{1}$ 
Akiko Fukuda,$^{2}$ 
Etsuo Segawa,$^{3}$ 
Iwao Sato$^{1}$ 
}\\ 
{\scriptsize $^1$ 
Oyama National College of Technology, 
}\\
{\scriptsize 
Oyama, Tochigi 323-0806, Japan
}\\
{\scriptsize $^{2}$ 
Department of Mathematical Sciences, Shibaura Institute of Technology, 
}\\
{\scriptsize 
Saitama 337-8570, Japan
}\\
{\scriptsize $^3$ 
Graduate School of Environment Information Sciences, 
}\\
{\scriptsize 
Yokohama National University, Yokohama, 240-8501, Japan
}\\
} 
\date{\empty }
\begin{document}
\maketitle

\par\noindent
\begin{small}
\par\noindent
{\bf Abstract}. 
Max-plus algebra is a kind of idempotent semiring over $\mathbb{R}_{\max}:=\mathbb{R}\cup\{-\infty\}$ with two operations 
$\oplus := \max$ and $\otimes := +$.
In this paper, we introduce a new model of a walk on one dimensional lattice on $\mathbb{Z}$, as an analogue of the quantum walk, over the max-plus algebra and we call it max-plus walk. 
In the conventional quantum walk, the summation of the $\ell^2$-norm of the states over all the positions is a conserved quantity.
In contrast, the summation of eigenvalues of state decision matrices is a conserved quantity in the max-plus walk.
Moreover, spectral analysis on the total time evolution operator is also given.


\end{small}

\setcounter{equation}{0}

\section{Introduction}
Ultradiscretization is a technique which transform a difference equation into piecewise linear equation and appears in the context of integrable systems \cite{TTMS1996}. 
It is based on the following formula.
\begin{eqnarray}\label{eq:UD}
\lim_{\epsilon \to +0 }
\epsilon\log(e^{A/\epsilon}+e^{B/\epsilon}) = \max\{A,B \}.
\end{eqnarray}
It is known that essential properties are preserved by the ultradiscretization for integrable systems \cite{Watanabe2018}. 
Since ultradiscrete equations can be written by using operations ``$\max$'' and ``$+$'',
they can be considered over max-plus algebra. 

In max-plus algebra, the sum of two elements
is their maximum and the product of two elements is their sum. 
This algebraic structure is known as idempotent semiring. 
Replacing maximum with minimum,  we get isomorphic min-plus algebra. 
Max-plus algebra has many analogies for conventional linear algebra \cite{BCOQ,Olsder,Schutter}. 
Intensive studies on quantum walks appear in the beginning of 2000's from the view point of quantum information e.g. \cite{Am,Meyer}. 
Now we can see overlaps of quantum walks to not only quantum information but also various kinds of research fields. 
One of the reason for the studies on quantum walks is not only the efficiency on quantum search algorithm but also 
application as a quantum simulator of quantum phenomena (on a quantum device) as envisioned by Feynman \cite{Fey} 
because of its universarity of the quantum computation~\cite{Childs}. 
For example, as a simulator of the Dirac equation e.g.,~\cite{MAMP,Sh} and its reference therein, 
while quantum graphs~\cite{HKSS, MMOS} which is a system of stationary Schr{\"o}dinger equations
on metric graphs, and topological insulator \cite{Kitagawa}. 
The time evolution of a discrete-time quantum walk is described by the iteration of a unitary operator 
on some Hilbert space generated by a discrete set. 
Then due to the unitary of the time evolution operator, we can define a distribution at each time step 
because the time evolution operator preserves the norm. 

From now on, we restrict ourselves to one-dimensional lattice. 
Let the standard basis of $\mathbb{C}^2$ be denoted by $|L\rangle=[1,0]^{\top}$, $|R\rangle=[0,1]^{\top}$. 
We put $\langle L|=(|L\rangle)^*$ and $\langle R|=(|R\rangle)^*$, respectively. 
Here $A^*$ is the conjugate and transpose of $A$. 
Let 
\begin{eqnarray*} 
{\bf H}= \left[ \begin{array}{cc} {\rm a}&{\rm b} \\ {\rm c}&{\rm d} \end{array} \right]
\end{eqnarray*}
be a unitary matrix on $\mathbb{C}^{2}$ and ${\bf P}=|L\rangle\langle L|{\bf H}$, ${\bf Q}=|R\rangle\langle R|{\bf H}$. 
The $\mathbb{C}^{2}$-valued amplitude $\varphi_{k}^{n}$, which is regarded as 
the $n$-th time evolution of the quantum walk at position $k\in\mathbb{Z}$, 
satisfies the following recursion equation: 
\begin{eqnarray}\label{eq:time_evolutionQW}
\varphi_{k}^{n}={\bf P}\varphi_{k+1}^{n-1}+{\bf Q}\varphi_{k-1}^{n-1}. 
\end{eqnarray}
The weights associated with left and right movings are ${\bf P}$ and ${\bf Q}$, respectively. 
This is a kind of quantum analogue of a random walk; ${\bf P}$ and ${\bf Q}$ with ``${\bf P}+{\bf Q}\in $ $2$-dimensional unitary matrix" 
corresponds to the probabilities associated with left and right movings $p$ and $q$ with ``$p+q=1$" in the random walk on $\mathbb{Z}$. 
Suppose that the initial state is 
\begin{eqnarray*}
\varphi_0^k=
\left\{
\begin{array}{c}
\phi,\quad k=0,\\
{\bm 0},\quad k\ne 0
\end{array}
\right.
\end{eqnarray*}
 with $\|\phi\|^2=1$. 
Then the state $\varphi_{k}^{n}$ is expressed by the summation over the matrix valued weights which are products of ${\bf P}$'s and ${\bf Q}$'s 
associated with $n$-length walks starting from the origin to position $k$ on $\mathbb{Z}$. 
Let us denote the above ($2$-dimensional) matrix valued weights be ${\bf A}_k^n$. 
Then $\varphi_k^n$ can be rewritten by 
\begin{eqnarray}
\varphi_{k}^{n}={\bf A}_{k}^{n}\phi.
\end{eqnarray}
In this paper we call ${\bf A}_{k}^{n}$ a state decision matrix. 
The explicit expression for ${\bf A}_{k}^{n}$ corresponding to $\binom{n}{n-k}p^{\frac{n+k}{2}}q^{\frac{n-k}{2}}$ for the random walk is obtained by Konno~(see \cite{Konno2008} and its reference therein). 
Let $\ell=(n-k)/2$ and $m=(n+k)/2$, which are the number of moving left and right during the $n$-step. Then 
\begin{eqnarray}\label{eq:QW_amplitude}
{\bf A}_{k}^{n} =
{\rm a}^{\ell}{\rm d}^{m}\sum_{r=1}^{\ell \wedge m} \left(\frac{{\rm bc}}{{\rm ad}}\right)^{r}
\binom{\ell-1}{r-1}
\binom{m-1}{r-1}
\left( \frac{\ell-r}{{\rm a}r}{\bf P}+\frac{m-r}{{\rm d} r}{\bf Q}+\frac{1}{{\rm c}}{\bf R}+\frac{1}{{\rm b}}{\bf S} \right), 
\end{eqnarray}
where ${\bf R}=|L\rangle \langle R|{\bf H}$ and ${\bf S}=|R\rangle\langle L|{\bf H}.$
\footnote{Note that ${\bf P}$, ${\bf Q}$, ${\bf R}$, ${\bf S}$ are the complete orthogonal basis on $2$-dimensional matrices 
in terms of the inner product $\mathrm{tr}(X^*Y)$ if $X\neq Y$, for any $2$-dimensional matrices $X$ and $Y$. 
}

The distribution at time $n$ and position $k$ is defined by $\mu_k^n=\|\varphi_k^n\|^2_{\mathbb{C}^2}$. 
We can regard it as the probability that a quantum walker is observed at time $n$ and position $k$. 
We can see some interesting properties of this quantum walk in the following weak limit theorem~(see \cite{Konno2008} and its reference therein):
if the initial state $\phi=[{\bm \alpha},{\bm \beta}]^{\top}$ satisfies $|{\bm \alpha}|=|{\bm \beta}|$ and $\mathrm{Re}({\rm a}{\bm \alpha}\overline{{\rm b}{\bm \beta}})=0$, then the scaled limit distribution is
\begin{eqnarray}
\lim_{n\to\infty}\sum_{k/n\leq u}\mu_{k}^{n}= \int_{-\infty}^{u} f_K(v;|{\rm a}|) dv \label{eq:Konno}
\end{eqnarray}
for any $u\in\mathbb{R}$, where 
\[ f_K(v;|{\rm a}|):=\frac{\boldsymbol{1}_{(-|{\rm a}|,|{\rm a}|)}(v)|{\rm b}|}{\pi(1-v^2)\sqrt{|{\rm a}|^2-v^2}}\]
and $\boldsymbol{1}_{(-|{\rm a}|,|{\rm a}|)}(v)$ is the indicator function on $(-|{\rm a}|,|{\rm a}|)$. 
This weak limit theorem corresponds to so called the central limit theorem of the random walk 
but the scaling order is quadratically larger than random walk and the limit density is far from the normal distribution. 
By the Fourier analysis on the quantum walk, the spectrum of the total time evolution operator $\mathcal{U}$ of 
this quantum walk can be described by $\mathrm{Spec}(\mathcal{U})=\{ {\Delta'}^{1/2} \exp(i\theta) \;|\; \cos\theta\in [-|{\rm a}|,|{\rm a}|]\}$, 
where $\Delta'=\det({\bf H})$. 
Moreover the spectral measure and its associated Laurent polynomials of this time evolution operator are obtained (see \cite{CGMV} and its reference therein).

In this paper, we try to obtain a new model of walk over the max-plus algebra which is analogous to the conventional quantum walk. 
We call it max-plus walk. 
To this end, first we need to determine a measurement process in the max-plus walks. 
Due to the unitarity of the time evolution of the quantum walks, the $\ell^2$-norm is preserved in the quantum walks. 
Then we can define a ``distribution" at each time step $n$ if the norm of the initial state is unit. 
Therefore in the quantum walks, a conserved quantity which are independent of the time iteration $n$, 
is the summation of the $\ell^2$-norm of the probability amplitudes and also the Frobenius norm of the state decision matrices over all the positions. 
So as an analogue of it, in this paper, we propose the conserved quantity of the max-plus walk
by the summation of the max-plus eigenvalues of the state decision matrices over all the positions. 
That is, letting $\lambda(A_k^n)$ be the eigenvalue of the state decision matrix $A_k^n$ on the max-plus algebra, then we define the 
conserved quantity as $\sum_{k}\lambda(A_k^n)$. 
We obtain a necessary and sufficient condition of the setting of the max-plus walk for conserving this value (see theorem \ref{thm:ConservedQuantity}). 
This condition in the max-plus walk corresponds to {\it unitarity} of the local quantum coin ${\bf H}$ in the quantum walks. 
Using this conservative property, we can define a quantity at each position corresponding to the probability distribution, 
and obtain its explicit expression for each time step.  
Under this conservative condition, we consider the spectral analysis on the time evolution operator on the infinite whole system
on the max-plus algebra, and obtain the eigenvalue and its eigenvector using a graph theoretical approach (see theorem \ref{thm:SpectralAnalysis}). 

This paper is organized as follows. 
In section 2,  we explain the definition and properties of the max-plus algebra. 
After the ultradiscretization (\ref{eq:UD}) of the recursion equation of the quantum walk (\ref{eq:time_evolutionQW}), 
the operation of the time evolution changes to the the max-plus algebra. 
Then in section 3, we devote to the amplitude of matrix valued weighted walks on the max-plus algebra namely the state decision matrix. 
On the time evolution of the max-plus walk, 
we obtain an explicit expression for the state decision matrix $A_k^n$ which corresponds to the state decision matrix 
${\bf A}_n^k$ of the quantum walk (\ref{eq:QW_amplitude}) (see theorem \ref{thm:MPW_amplitude}). 
Interestingly, although we just take a ultradiscretization of the discrete recursion equation describing the time evolution of the quantum walk, 
we obtain $A_k^n$, which is the $n$-th iteration of the max-plus walk at position $k$, is also nothing but the ultradiscretization of ${\bf A}_k^n$.  
Then in section 4, we propose the conserved quantities of the max-plus walk corresponding to the $\ell^2$ conservation of the quantum walks and show a necessary and sufficient condition of the setting of the max-plus walk. 
In section 5, under the condition, we obtain the eigenvalue and its eigenvector on the max-plus algebra. 
In section 6, we finally give concluding remarks and discussions. 

\section{Preliminaries on max-plus algebra}

Let $\mathbb{R}_{\max}=\mathbb{R}\cup \{-\infty \}$ be 
the set of all real numbers $\mathbb{R}$ together with 
an extra element $-\infty$ expressing negative infinity. 
We define two operations, addition $\oplus$ and multiplication 
$\otimes$, in $\mathbb{R}_{\max}$ in terms of conventional 
operations by
\begin{eqnarray*}
a\oplus b =\max\{a,b \}, \quad 
a\otimes b =a+b \quad (a,b\in \mathbb{R}_{\max}).
\end{eqnarray*}
Remark here that, in this paper, we sometimes use $\max$ and $+$ instead of $\oplus$ and $\otimes$, respectively, for convenience. 

Then $(\mathbb{R}_{\max},\oplus,\otimes)$ is a commutative semiring 
called max-plus algebra. 
Here, $-\infty$ is the identity element for addition and $0$ is 
the identity element for multiplication:  
we express as $\varepsilon=-\infty$ and $e=0$, respectively. 
For details about the max-plus algebra refer to \cite{BCOQ}.
\par
Let $\mathbb{R}_{\max}^{m\times n}$ be the set of $m\times n$ 
matrices whose entries are in $\mathbb{R}_{\max}$. 
The set $\mathbb{R}_{\max}^{n\times 1}$ of vectors is abbreviated as 
$\mathbb{R}_{\max}^{n}$.
The arithmetic operations on vectors and matrices 
are defined as those in the conventional linear algebra.
For max-plus matrices 
$A=([A]_{ij}), B=([B]_{ij})\in \mathbb{R}_{\max}^{m\times n}$, 
we define the matrix sum 
$A\oplus B =([A\oplus B]_{ij})\in \mathbb{R}_{\max}^{m\times n}$ by 
\begin{eqnarray*}
[A\oplus B]_{ij}= [A]_{ij} \oplus [B]_{ij}.
\end{eqnarray*}
For max-plus matrices 
$A=([A]_{ij})\in \mathbb{R}_{\max}^{m\times p}$ and 
$B=([B]_{ij})\in \mathbb{R}_{\max}^{p\times n}$, 
we define the matrix multiplication 
$A\otimes B=([A\otimes B]_{ij})\in \mathbb{R}_{\max}^{m\times n}$ by
\begin{eqnarray*}
[A\otimes B]_{ij}= \bigoplus_{k=1}^{p}
[A]_{ik} \otimes [B]_{kj}.
\end{eqnarray*}
For a max-plus matrix $A=([A]_{ij})\in \mathbb{R}_{\max}^{m\times n}$ 
and a scalar $\alpha \in \mathbb{R}_{\max}$, 
we define the scalar multiplication 
$\alpha \otimes A=([\alpha \otimes A]_{ij})\in \mathbb{R}_{\max}^{m\times n}$
by 
\begin{eqnarray*}
[\alpha\otimes A]_{ij}= \alpha \otimes [A]_{ij}.
\end{eqnarray*}
The matrix $I_n\in \mathbb{R}_{\max}^{n\times n}$ whose diagonal 
entries are $e$ and other entries are $\varepsilon$ is the identity matrix.
\par
For a matrix $A=([A]_{ij})\in \mathbb{R}_{\max}^{n\times n}$, 
we define the tropical determinant of $A$ by 
\begin{eqnarray*}
\text{tropdet}(A)=
\bigoplus_{\sigma \in S_n} [A]_{1\sigma(1)}\otimes 
[A]_{2\sigma(2)}\otimes \cdots \otimes [A]_{n\sigma(n)}
 \end{eqnarray*}
where $S_n$ denotes the symmetric group of order $n$. 
\par
Next, we review the relationships between the eigenvalue problem 
of min-plus matrices and the corresponding digraphs.
A digraph is a orderd pair $G=(V,E)$ where $V$ is a nonempty finite 
set and $E\subseteq V\times V$. 
Elements $v\in V$ and $e\in E$ are called vertices and
edges, respectively. 
A sequence $P=(v_1,v_2,\dots,v_s)$ of vertices is called a $v_1$--$v_s$ path if 
$(v_i,v_{i+1})\in E$ for all $i=1,2,\dots,s-1$.
The number $s-1$ is called the length of $P$ and is denoted by $\ell(P)$.
A $v_1$--$v_s$ path is called a circuit if $v_1=v_s$. 
A digraph $G$ is called a strongly connected if there
exists a $u$--$v$ path for all vertices $u,v\in V$. 
A weighted digraph is the tuple $G=(V,E,w)$ where 
$(V,E)$ is a digraph and $w$ is a real function on $E$.
If $P=(v_1,v_2,\dots,v_s)$ is a path in the weighted 
digraph $G=(V,E,w)$, 
then the weight of $P$ is 
$w(P)=w((v_1,v_2))+w((v_2,v_3))+\cdots +w((v_{s-1},v_s))$.
For the circuit $C$ in $G=(V,E,w)$, 
the average weight $\text{ave}(C)$ of $C$ is defined as 
\begin{eqnarray*}
\text{ave}(C)=\frac{w(C)}{\ell(C)}.
\end{eqnarray*}
Let $G=(V,E,w)$ be a weighted digraph with $V=\{1,2,\dots,n\}$, 
then the weighted adjacency matrix 
$A(G)=([A]_{ij})\in \mathbb{R}_{\max}^{n\times n}$
is defined by 
\begin{eqnarray*}
[A]_{ij}=
\left\{
\begin{array}{ll}
w((i,j)) & \text{if } (i,j)\in E, \\
\varepsilon & \text{otherwise}.
\end{array}
\right.
\end{eqnarray*}
Conversely, for any max-plus matrix 
$A\in \mathbb{R}_{\max}^{n\times n}$, 
there exists a weighted digraph whose weghted matrix is $A$.
We denote such a weighted digraph by $G(A)$.
\par
Let $G=(V,E,w)$ be a weighted digraph and $A(G)$ be the weighted
adjacency matrix of $G$. 
For $i,j=1,2,\dots,n$, we denote by $a_{ij}^*$ the maximum
value of weights of all $i$--$j$ paths in $G$.
We set $a_{ij}^* =\varepsilon$ if there exists no $i$--$j$ paths.
The maximum weight matrix $A^*(G)=A^* \in \mathbb{R}_{\max}^{n\times n}$
is defined by $A^* =(a_{ij}^*)$.

\begin{proposition}[\cite{BCOQ}]
For a max-plus matrix $A\in \mathbb{R}_{\max}^{n\times n}$, 
the maximum weight matrix $A^*$ of $G(A)$ can be computed by 
the following power series:
\begin{eqnarray*}
A^* =I\oplus A \oplus A^{\otimes 2} \oplus \cdots.
\end{eqnarray*}
Moreover, if the weighted digraph $G(A)$ has no positive circuits, then
\begin{eqnarray*}
A^* = I\oplus A \oplus A^{\otimes 2} \oplus \cdots \oplus A^{\otimes n-1}.
\end{eqnarray*}
\end{proposition}

For a max-plus matrix $A\in \mathbb{R}_{\min}^{n\times n}$, 
if there exist $\lambda \in \mathbb{R}_{\max}$ and 
$\bm{x}\in \mathbb{R}_{\max}^{n}\setminus 
\{(\varepsilon,\varepsilon,\dots,\varepsilon)^{\top}\}$ satisfying 
\begin{eqnarray*}
A\otimes \bm{x} = \lambda \otimes \bm{x}, 
\end{eqnarray*}
then $\lambda$ and $\bm{x}$ are called an eigenvalue and 
its corresponding eigenvector, respectively.

The eigenvalues and its corresponding eigenvector are 
shown in \cite{BCOQ} to have interesting relationships 
with circuits and maximum weight matrices in the weighted digraph.

\begin{proposition}[\cite{BCOQ}]\label{prop:mp_eig}
Let $A\in \mathbb{R}_{\max}^{n\times n}$ be a max-plus matrix.
If $A$ has an eigenvalue $\lambda\not=\varepsilon$, 
then there exists a circuit in the weighted digraph $G(A)$ 
whose average weight is equal to $\lambda$.
In particular, if $G(A)$ is strongly connected, 
there is only one eigenvalue of $A$ 
which is equal to the maximum average weight of circuits.
\end{proposition}

\begin{proposition}[\cite{BCOQ}]
\label{prop:mp_evec}
Let $\lambda \not=\varepsilon$ be an eigenvalue of 
$A\in \mathbb{R}_{\max}^{n\times n}$ and $\bm{x}_i$ be the 
$i$th column of $(-\lambda\otimes A)^*$.
If a vertex $i$ is contained in the circuit whose average 
weight is equal to $\lambda$, 
then $\bm{x}_i$ is eigenvector corresponding to $\lambda$ of $A$.
\end{proposition}

\section{Max-plus walk: a walk on the max-plus algebra}
A max-plus walk considered here is determined as follows.
Let $k$ and $n$ be a position on one dimensional lattice on $\mathbb{Z}$ and discrete time, respectively.
Then a max-plus vector $\psi_{k}^{n}\in \mathbb{R}_{\max}^{2}$ is determined by the following evolutionary equation:
\begin{eqnarray}\label{mpqw_evolution}
\psi_k^{n} =(P\otimes \psi_{k+1}^{n-1}) \oplus 
(Q\otimes \psi_{k-1}^{n-1}),
\end{eqnarray}
where $P$ and $Q$ are the $2\times 2$ max-plus matrices with $a,b,c,d \in \mathbb{R}_{\max}\setminus \{\varepsilon \}$
\begin{eqnarray*}
P=
\left[
\begin{array}{cc}
a & b \\
\varepsilon & \varepsilon
\end{array}
\right]
,\quad 
Q=
\left[
\begin{array}{cc}
\varepsilon & \varepsilon \\
c & d
\end{array}
\right].
\end{eqnarray*}
We define $H=P\oplus Q$.
We set the initial state $\psi_{k}^{0}$ as
\begin{eqnarray*}
\psi_{k}^{0} = 
\left\{
\begin{array}{l}
\left[
\begin{array}{c}
\alpha\\\beta
\end{array}
\right],\quad k=0,\vspace{5pt}\\
\left[
\begin{array}{c}
\varepsilon\\ \varepsilon
\end{array}
\right],\quad k\ne 0,
\end{array}
\right.
\end{eqnarray*}
where at least one of $\alpha$ and $\beta$ is not equal to $\varepsilon$.
%
%
The max-plus walk with initial state $\psi_{0}^{0}$ is expressed by
\begin{eqnarray}\label{mpqw_evolution2}
\psi_k^{n} =(P\otimes \psi_{k+1}^{n-1}) \oplus 
(Q\otimes \psi_{k-1}^{n-1})
=A_k^n\otimes \psi_0^0,
\end{eqnarray}
where we call $A_k^n$ state decision matrix.
For example, $A_{1}^{3}$ can be written down as
\begin{eqnarray}
A_{1}^{3} = 
(Q\otimes Q\otimes P)
\oplus
(Q\otimes P\otimes Q)
\oplus
(P\otimes Q\otimes Q).
\label{A_1^3}
\end{eqnarray}
We here introduce two max-plus matrices $R,S\in \mathbb{R}_{\max}^{2\times 2}$ as 
\begin{eqnarray*}
R=
\left[
\begin{array}{cc}
c & d \\
\varepsilon & \varepsilon
\end{array}
\right]
,\quad 
S=
\left[
\begin{array}{cc}
\varepsilon& \varepsilon\\
a & b
\end{array}
\right].
\end{eqnarray*}
Then the matrices $P,Q,R$ and $S$ have the relationship concerning the products shown in table \ref{table:products}.
\begin{table}
\caption{Products of the matrices $P$, $Q$, $R$ and $S$}\label{table:products}
\begin{center}
\begin{tabular}{c|cccc} 
    & $P$      & $Q$      & $R$      & $S$ \\ \hline
$P$ & $a\otimes P$ & $b\otimes R$ & $a\otimes R$ & $b\otimes P$ \\
$Q$ & $c\otimes S$ & $d\otimes Q$ & $c\otimes Q$ & $d\otimes S$ \\
$R$ & $c\otimes P$ & $d\otimes R$ & $c\otimes R$ & $d\otimes P$ \\
$S$ & $a\otimes S$ & $b\otimes Q$ & $a\otimes Q$ & $b\otimes S$ 
\end{tabular}
\end{center}
\end{table}
Using table \ref{table:products},  (\ref{A_1^3}) can be written by
\begin{eqnarray*}
A_{1}^{3} &= (b+c)\otimes Q\oplus (b+d)\otimes R\oplus (c+d)\otimes S\\
&=
\left[
\begin{array}{cc}
b+c+d & b+2d \\
\max\{b+2c,a+c+d\} & b+c+d 
\end{array}
\right].
\end{eqnarray*}
Let $\ell$ and $m$ be the number of moving left and right, respectively, during $n$-step.
Then we have the following theorem.
\begin{theorem}\label{thm:MPW_amplitude}
The state decision matrix $A_k^{n}$ can be written as follows.
In cases of $k=-n$ and $n$, $A_{k}^{n}$ are given by
\begin{eqnarray}
A_{-n}^{n}=(n-1)a\otimes P,\quad A_{n}^{n}=(n-1)d\otimes Q.\label{A_boundary}
\end{eqnarray}
In cases of $k=-n+2,-n+4,\dots,n-4,n-2$, then $A_k^{n}$ are given by
\begin{align}
A_k^{n} &=
\bigoplus_{r=1}^{(\ell-1) \wedge m}
\left\{
(\ell -r-1)a +rb+rc+(m-r)d
\right\} \otimes P \notag \\
&  \oplus \bigoplus_{r=1}^{\ell \wedge (m-1)}
\left\{
(\ell -r)a+rb +rc +(m-r-1)d
\right\} \otimes Q \notag \\
& \oplus \bigoplus_{r=1}^{\ell \wedge m}
\left\{
(\ell -r)a+rb+(r-1)c+(m-r)d
\right\} \otimes R \notag\\
&  \oplus \bigoplus_{r=1}^{\ell \wedge m}
\left\{
(\ell -r)a +(r-1)b +rc +(m-r)d
\right\}
\otimes S,\label{eq:SDM}
\end{align}
where $a\wedge b = \min\{a,b\}$.
\end{theorem}
\begin{proof}
From table \ref{table:products}, 
it is obvious that (\ref{A_boundary}) holds.
In cases of $k=-n+2,-n+4,\dots,n-4,n-2$,
the set of $n$-length path from the origin to position $k=m-\ell$ can be decomposed into the following four cases:
letting the first and last choices of the directions be described by $(X,Y)$, where $X,Y\in \{\mathrm{left},\mathrm{right}\}$, then
we decompose (i) $(\mathrm{left},\mathrm{left})$
(ii) $(\mathrm{right},\mathrm{right})$ (iii) $(\mathrm{right},\mathrm{left})$ and 
(iv) $(\mathrm{left},\mathrm{right})$. 
The corresponding matrix valued weight of path is given as follows. 
\begin{align*}
{\mathrm (i)}~&
\xi_{P}(w_{2r+1},w_{2r},\dots,w_{1})\\
&=\overbrace{P\otimes \cdots \otimes P}^{w_{2r+1}}
\otimes 
\overbrace{Q\otimes \cdots \otimes Q}^{w_{2r}}
\otimes
\overbrace{P\otimes \cdots \otimes P}^{w_{2r-1}}\otimes 
\dots \otimes \overbrace{Q\otimes \cdots \otimes Q}^{w_2}
\otimes \overbrace{P\otimes \cdots \otimes P}^{w_1}\\
&\qquad\qquad\qquad\qquad\qquad\qquad\qquad\qquad\qquad\qquad\qquad\qquad\qquad {\rm for}~r=1,2,\dots, \ell-1\wedge m\\ 
{\mathrm (ii)}~&
\xi_{Q}(w_{2r+1},w_{2r},\dots,w_{1})\\
& =\overbrace{Q\otimes \cdots \otimes Q}^{w_{2r+1}}
\otimes 
\overbrace{P\otimes \cdots \otimes P}^{w_{2r}}
\otimes
\overbrace{Q\otimes \cdots \otimes Q}^{w_{2r-1}}\otimes 
\dots \otimes \overbrace{P\otimes \cdots \otimes P}^{w_2}
\otimes \overbrace{Q\otimes \cdots \otimes Q}^{w_1}\\
&\qquad\qquad\qquad\qquad\qquad\qquad\qquad\qquad\qquad\qquad\qquad\qquad\qquad {\rm for}~r=1,2,\dots, \ell\wedge m-1\\ 
{\rm (iii)}~&
\xi_{R}(w_{2r},\dots,w_{1})\\
& =\overbrace{P\otimes \cdots \otimes P}^{w_{2r}}
\otimes 
\overbrace{Q\otimes \cdots \otimes Q}^{w_{2r-1}}
\otimes\cdots \otimes 
\overbrace{P\otimes \cdots \otimes P}^{w_{2}}
\otimes \overbrace{Q\otimes \cdots \otimes Q}^{w_1}\\
&\qquad\qquad\qquad\qquad\qquad\qquad\qquad\qquad\qquad\qquad\qquad\qquad\qquad {\rm for}~r=1,2,\dots, \ell\wedge m\\ 
{\rm (iv)}~&
\xi_{S}(w_{2r},\dots,w_{1})\\
& =\overbrace{Q\otimes \cdots \otimes Q}^{w_{2r}}
\otimes
\overbrace{P\otimes \cdots \otimes P}^{w_{2r-1}}\otimes 
\dots \otimes \overbrace{Q\otimes \cdots \otimes Q}^{w_2}
\otimes \overbrace{P\otimes \cdots \otimes P}^{w_1}\\
&\qquad\qquad\qquad\qquad\qquad\qquad\qquad\qquad\qquad\qquad\qquad\qquad\qquad {\rm for}~r=1,2,\dots, \ell\wedge m 
\end{align*}
We first consider the case (i). 
From table \ref{table:products}, we have
\begin{align*}
\xi_{P} & (w_{2r+1},w_{2r},\dots,w_{1})\\
&= 
(w_{2r+1}-1)a\otimes P\otimes 
(w_{2r}-1)d\otimes Q\otimes 
\cdots
\otimes
(w_{2}-1)d\otimes Q \otimes 
(w_{1}-1)a\otimes P\\
&= (w_{2r+1} +w_{2r-1} + \cdots + w_{1} - (r+1) )a
\otimes
(w_{2r} +w_{2r-2} + \cdots + w_{2} - r )d\\
&   \qquad\qquad\qquad\qquad\qquad\qquad\qquad\qquad\qquad\qquad\qquad\qquad  \otimes 
P\otimes Q\otimes \cdots \otimes  Q\otimes  P\\
&=(\ell - r-1 )a
\otimes
(m - r)d\otimes
rb\otimes (r-1) c\otimes R \otimes P\\
&=(\ell - r-1 )a
\otimes rb\otimes
 rc \otimes (m - r)d \otimes P\\
&=\{(\ell - r-1 )a
+ rb +
 rc + (m - r)d\} \otimes P.
\end{align*}
For (ii), (iii) and (iv), in similar ways, we have 
\begin{align*}
\xi_{Q}(w_{2r+1},w_{2r},\dots,w_{1})&=\left\{
(\ell -r)a+rb +rc +(m-r-1)d
\right\} \otimes Q,\\
\xi_{R}(w_{2r},w_{2r-1},\dots,w_{1})&=
\left\{
(\ell -r)a+rb+(r-1)c+(m-r)d
\right\} \otimes R,\\
\xi_{S}(w_{2r},w_{2r-1},\dots,w_{1})
&= \left\{
(\ell -r)a +(r-1)b +rc +(m-r)d
\right\}
\otimes S.
\end{align*}
Remark that $A_{k}^{n}$ can be expressed by
\begin{eqnarray*}
&    A_{k}^{n}=\bigoplus_{r=1}^{(\ell-1)\wedge m}\bigoplus_{w_{1},\dots,w_{2r+1}}\xi_{P}(w_{2r+1},\dots,w_{1})
\oplus
\bigoplus_{r=1}^{\ell\wedge (m-1)}\bigoplus_{w_{1},\dots,w_{2r+1}}\xi_{Q}(w_{2r+1},\dots,w_{1})\\
&   \qquad \oplus
\bigoplus_{r=1}^{\ell\wedge m}\bigoplus_{w_{1},\dots,w_{2r}}\xi_{R}(w_{2r},\dots,w_{1})
\oplus
\bigoplus_{r=1}^{\ell \wedge m}\bigoplus_{w_{1},\dots,w_{2r}}\xi_{S}(w_{2r},\dots,w_{1}).
\end{eqnarray*}
Since $\xi_{P}, \xi_{Q}, \xi_{R}$ and $\xi_{S}$ is not depend on $w_{k}$,
 by inserting (i)--(iv) into the above right hand side, we obtain the desired conclusion.
\end{proof}
It is remarkable here that $A_{k}^{n}$ is just an ultradiscretization of ${\bf A}_{k}^{n}$ in (\ref{eq:QW_amplitude}).

\section{Conserved quantities of the max-plus walk}
In the conventional quantum walk, 
the summation of the $\ell^2$-norm of the state vector $\varphi_{k}^{n}$ is conserved quantities with respect to the discrete time $n$. 
While, in the conventional quantum walk, it is difficult to calculate the eigenvalues of the state decision matrix (\ref{eq:QW_amplitude}), in the  max-plus walk, we can easily obtain them. 
In this section we discuss what corresponds to the conserved quantities in the max-plus walk. 
We approach the question by calculating eigenvalues of $A_{k}^{n}$ (\ref{eq:SDM}).

From (\ref{A_boundary}), $A_{k}^{n}$ for $k=-n$ and $n$ are as follows.
\begin{eqnarray*}
A_{-n}^{n} = (n-1)a
\otimes
\left[
\begin{array}{cc}
a&b\\
\varepsilon & \varepsilon
\end{array}
\right]
,\quad 
A_{n}^{n} = (n-1)d
\otimes
\left[
\begin{array}{cc}
\varepsilon & \varepsilon\\
c&d
\end{array}
\right].
\end{eqnarray*}
From proposition \ref{prop:mp_eig}, 
it is easy to see that eigenvalues of $A_{-n}^{n}$ and $A_{n}^{n}$ are $na$ and $nd$, respectively.
Next we consider the case of  $k=-n+2,-n+4,\dots,n-4,n-2$. 
Let us define $\Delta := (b+c) - (a+d)$, 
then $A_{k}^{n}$ (\ref{eq:SDM}) is given as 
\begin{align*}
A_k^{n} &=
\bigoplus_{r=1}^{(\ell-1) \wedge m}
\left\{
(\ell -1)a +md +r \Delta
\right\} \otimes P   \oplus \bigoplus_{r=1}^{\ell \wedge (m-1)}
\left\{
\ell a+(m-1)d + r  \Delta
\right\} \otimes Q \\
&\qquad\qquad\qquad\qquad \oplus \bigoplus_{r=1}^{\ell \wedge m}
\left\{
\ell a+md-c + r  \Delta
\right\} \otimes R  \oplus \bigoplus_{r=1}^{\ell \wedge m}
\left\{
\ell a+md -b + r  \Delta
\right\}
\otimes S\\
&=
(\ell a+md) \otimes \left[
\bigoplus_{r=1}^{(\ell-1) \wedge m}
(-a + r \Delta)\otimes P  
\right.
 \oplus \bigoplus_{r=1}^{\ell \wedge (m-1)}
(- d + r  \Delta ) \otimes Q \\
&\qquad\qquad\qquad\qquad\qquad\qquad\qquad\qquad \oplus \bigoplus_{r=1}^{\ell \wedge m}
(- c + r  \Delta ) \otimes R \left. \oplus \bigoplus_{r=1}^{\ell \wedge m}
(-b + r  \Delta)
\otimes S
\right].
\end{align*}
If $\Delta\geq 0$, then we have 
\begin{align*}
A_k^{n} &=
(\ell a+md) \otimes [
(-a + \min\{\ell-1,m \}\Delta) \otimes P  
\oplus 
(- d + \min\{\ell,m-1 \}  \Delta) \otimes Q \\
&\qquad\qquad\qquad\qquad\qquad\qquad\qquad \oplus 
(- c + \min\{\ell, m \}  \Delta ) \otimes R  \oplus
(-b + \min\{\ell,m \} \Delta )
\otimes S]\\
&=
(\ell a+md) \otimes \left\{
\left[
\begin{array}{cc}
\min\{\ell-1,m\}\Delta & -a+b+\min\{\ell-1,m\}\Delta\\
\varepsilon &\varepsilon \\
\end{array}
\right]\right.
\\
&\qquad\qquad\qquad\quad \oplus 
\left[
\begin{array}{cc}
\varepsilon &\varepsilon \\
c-d+\min\{\ell,m-1\}\Delta & \min\{\ell,m-1\}\Delta\\
\end{array}
\right]
 \\
&\qquad\qquad\qquad\quad \oplus 
\left[
\begin{array}{cc}
\min\{\ell,m\}\Delta &-c+d + \min\{\ell,m\}\Delta \\
\varepsilon &\varepsilon \\
\end{array}
\right]
 \\
&\qquad\qquad\qquad\quad \left. \oplus
\left[
\begin{array}{cc}
\varepsilon &\varepsilon \\
a-b+\min\{\ell,m\}\Delta &\min\{\ell,m\}\Delta \\
\end{array}
\right]
\right\}\\
&=(\ell a +md)\otimes
\left[
\begin{array}{cc}
\min\{\ell,m\}\Delta &-c+d +\min\{\ell,m+1\}\Delta \\
c-d+\min\{\ell,m-1\}\Delta &\min\{\ell,m\}\Delta \\
\end{array}
\right],
\end{align*}
where, in the last equality, $(1,2)$-entry is obtained by
\begin{align*}
(-a+b+\min\{\ell-1,m\}\Delta) &
\oplus 
(-c+d+\min\{\ell,m\}\Delta) \\
&=(\Delta -c+d+\min\{\ell-1,m\}\Delta)
\oplus 
(-c+d+\min\{\ell,m\}\Delta)\\
&=(-c+d+\min\{\ell,m+1\}\Delta)
\oplus 
(-c+d+\min\{\ell,m\}\Delta) \\
&= -c+d+\min\{\ell,m+1\}\Delta.
\end{align*}
Then, from proposition \ref{prop:mp_eig}, we obtain eigenvalues of $A_{k}^{n}$ for $k=-n+2,-n+4,\dots,n-4,n-2$ as 
\begin{align*}
\lambda({A_{k}^{n}}) &= 
\ell a + md +\dfrac{\min\{\ell,m+1\}\Delta + \min\{\ell,m-1\}\Delta}{2} \\
&= \ell a + md + \min \left\{ \ell, m, \dfrac{\ell +m-1}{2}  \right\}\Delta.
\end{align*}
If $\Delta<0$, then we have
\begin{align*}
A_k^{n} &=
(\ell a+md) \otimes [
(-a + \Delta) \otimes P
\oplus 
(- d +  \Delta) \otimes Q\\
 &\qquad \qquad \qquad \qquad \oplus 
(- c + \Delta ) \otimes R 
\oplus
(-b + \Delta )
\otimes S]\\
 &=
(\ell a+md) \otimes \left\{
\left[
\begin{array}{cc}
\Delta & -a+b+\Delta\\
\varepsilon &\varepsilon \\
\end{array}
\right]
 \oplus 
\left[
\begin{array}{cc}
\varepsilon &\varepsilon \\
c-d+\Delta & \Delta\\
\end{array}
\right]\right.
 \\
&\left.\qquad\qquad\qquad\quad  \oplus 
\left[
\begin{array}{cc}
\Delta &-c+d + \Delta \\
\varepsilon &\varepsilon \\
\end{array}
\right]
 \oplus
\left[
\begin{array}{cc}
\varepsilon &\varepsilon \\
a-b+\Delta &\Delta \\
\end{array}
\right]
\right\}\\
&=
(\ell a + md)\otimes
\left[
\begin{array}{cc}
\Delta &-c+d +\Delta \\
c-d&\Delta \\
\end{array}
\right].
\end{align*}
Then, eigenvalue of $A_{k}^{n}$ is 
\begin{eqnarray*}
\lambda(A_{k}^{n}) = \ell a + md + \dfrac{\Delta}{2}.
\end{eqnarray*}
To sum up, eigenvalues of $A_{k}^{n}$ are given as 
\begin{eqnarray*}
\lambda(A_{k}^{n}) = 
\left\{
\begin{array}{l}
 \ell a + md + \min \left\{ \ell, m, \dfrac{\ell +m-1}{2}  \right\}\Delta,\quad {\text{if}}~~ \Delta\geq 0,\\
\ell a + md + \dfrac{\Delta}{2}, \quad  {\text{if}}~~\Delta<0. \\
\end{array}
\right.
\end{eqnarray*}
Here we introduce the following assumption into the entries of the matrix $H$.
\begin{eqnarray}
({\rm A})~
\left\{
\begin{array}{c}
a+d=0, \\
b+c=0. \\
\end{array}
\right.
\label{C1}
\end{eqnarray}
Under the assumption ({\rm A}), 
noting that $\Delta =0$ and $\ell a +md =(\ell -m)a=-ka$,
the state decision matrix $A_{k}^{n}$ is given as follows. 
\begin{eqnarray}
    A_k^n = \label{A_k^n}
\left\{
\begin{array}{l}
\left[
\begin{array}{cc}
na & (n-1)a+b \\
\varepsilon& \varepsilon
\end{array}
\right],\quad 
k=-n,  \vspace{5pt} \\ \vspace{5pt}
\left[
\begin{array}{cc}
-ka & (-k-1)a+b \\
(-k+1)a-b & -ka
\end{array}
\right],\quad 
k=-n+2,-n+4,\dots,n-2,
\\
\left[
\begin{array}{cc}
\varepsilon& \varepsilon\\
(-n+1)a-b & -na 
\end{array}
\right], \quad k=n.
\end{array}
\right.
\end{eqnarray}
From proposition \ref{prop:mp_eig}, eigenvalues of $A_{k}^{n}$ are 
\begin{eqnarray*}
\lambda(A_k^n) =
\left\{
\begin{array}{l}
na,\quad  k=-n ,  \\
-ka,\quad  
k=-n+2,-n+4,\dots,n-2 , 
\\
-na,\quad  k=n. 
\end{array}
\right.
\end{eqnarray*}
Therefore, for any position $k$ eigenvalues can be written as
\begin{eqnarray}
\lambda(A_k^n) =-ka,\quad k=-n, -n+2,\dots,n-2,n.
\label{A_eigenvalue}
\end{eqnarray}
Summation of eigenvalues for all position $k$ is
\begin{eqnarray*}
\sum_{k} \lambda(A_{k}^{n}) = \sum_{k}-ka =0,
\end{eqnarray*}
which is conserved quantities with respect to the discrete time $n$.

In fact, the assumption (A) is the necessary and sufficient condition for that the summation of $\lambda(A_{k}^{n})$ for $k$ is conserved quantities, which is discussed in the following theorem.
\begin{theorem}\label{thm:ConservedQuantity}
In the max-plus walk, the summation of $\lambda(A_{k}^{n})$ for $k$ is conserved quantities with respect to $n$, if and only if
the assumption (A) holds for the matrix $H$.
Then, the conserved quantities are 0.
\end{theorem}
\begin{proof}
The necessity part of the theorem is already proved. 
We give a proof for the sufficiency of the theorem.
In case that $\Delta<0$, from $\ell = \dfrac{n-k}{2}$ and $m=\dfrac{n+k}{2}$, 
then the summation of eigenvalues of $A_k^{n}$ is 
\begin{align*} 
\sum_{k} \lambda(A_k^{n}) &= \sum_{k} \left(\dfrac{n-k}{2} a +\dfrac{n+k}{2}d+\dfrac{\Delta}{2}\right)\\
&=\dfrac{1}{2}\sum_{k}\left\{(-a+d)k + (a+d)n+\Delta\right\}\\
&=\dfrac{1}{2}\sum_{k}\left\{(a+d)n+\Delta\right\}\\
&=\dfrac{1}{2}\{(2n+1)(a+d)n + (2n+1)\Delta\}.
\end{align*}
Since $\Delta<0$, there exists no $\Delta$ such that $\sum_{k} \lambda(A_k^{n})$ yields to be constant. 
In case that $\Delta\geq 0$,
\begin{align*} 
\sum_{k} \lambda(A_k^{n}) &= \sum_{k}
\left\{
\dfrac{n-k}{2} a +\dfrac{n+k}{2}d+
\min
\left\{
\dfrac{n-k}{2},\dfrac{n+k}{2},\dfrac{n-1}{2}
\right\}
\Delta 
\right\}
\\
&=\dfrac{1}{2}\sum_{k}\left\{ (-a+d)k + (a+d)n + \min\{n-k,n+k,n-1\}\Delta\right\}\\
&=\dfrac{1}{2}\sum_{k}\left\{(a+d)n + \min\{n-k,n+k,n-1\}\Delta \right\}\\
&=\dfrac{1}{2}\left\{
n(2n+1)(a+d) + \left[\sum_{k=-n}^{-1}(n+k) + (n-1)
+ \sum_{k=1}^{n}(n-k)\right]\Delta
\right\}\\
&=\dfrac{1}{2}\left\{
n(2n+1)(a+d) + (n^{2}-1)\Delta\right\}   
\end{align*}
If $\Delta>0$, then $\sum_{k} \lambda(A_{k}^{n})$ is not constant.
When $\Delta=0$, we have
\begin{eqnarray*}
\sum_{k} \lambda(A_{k}^{n})=\dfrac{n(2n+1)}{2}(a+d).
\end{eqnarray*}
If $\sum_{k} \lambda(A_k^{n}) $ is constant, then we have $\Delta=0$ and $a+d=0$.
Therefore, $a+d=0$ and $b+c=0$, which are the assumption (A), must hold. 
\end{proof}
It is noted here that
the assumption (A) is a max-plus analogue of the property that
determinants of unitary matrices are 1 in linear algebra. 
Namely, 
\begin{eqnarray*}
{\rm tropdet} (H) =(a\otimes d) \oplus (b\otimes c) = e=0.
\end{eqnarray*}

\section{Spectral analysis on the total time evolution operator}
In this section, we consider the time evolution of the whole system of the max-plus walk.
Here we also impose the assumption (A). 
Now we introduce the infinite matrix $\cal{A}$ and the infinite vector $\Psi^{n}$ as follows.
\begin{eqnarray}
&\hspace{-80pt}
[-1]\hspace{5pt}[0]\hspace{6pt}[1]\notag\\
&{\cal{A}}=
 \begin{array}{c}
\left[-1 \right] \\
\left[0 \right]\\
\left[1 \right]
\end{array}
\hspace{-10pt}
\left[
\begin{array}{ccccc}
\ddots & \vdots & \vdots & \vdots &  \\
\cdots  & \cal{E} & P & \cal{E} & \cdots \\
\cdots  & Q & \cal{E} & P & \cdots \\
\cdots  & \cal{E} & Q & \cal{E} & \cdots \\
  & \vdots & \vdots & \vdots &  \ddots
\end{array}
\right]
,\quad 
\Psi^{n} = 
\begin{array}{c}
\\
\left[-1\right]\\
\left[0\right]\\
\left[1\right]\\
\\
\end{array}
\hspace{-10pt}
\left[
\begin{array}{c}
\vdots\\
\psi_{-1}^{n}\\
\psi_{0}^{n}\\
\psi_{1}^{n}\\
\vdots\\
\end{array}
\right],\label{calA}
\end{eqnarray}
where ${\cal{E}}$ is the matrix (or the vector) whose all entries are $\varepsilon$, and the symbol $[i]$ denotes the two rows (or columns) with indices $(i;{\rm L},i;{\rm R})$, for example, 
\begin{eqnarray*}
&\hspace{109pt}
{\scriptsize
\begin{array}{c}
0;{\rm L}\hspace{5pt}0;{\rm R}
\end{array}}\\
&[{\cal{A}}]_{[-1][0]} = P =
{\scriptsize
\begin{array}{c}
-1;{\rm L}\\[6pt]
-1;{\rm R}\\
\end{array}
}
\hspace{-8pt}
\left[
\begin{array}{cc}
a&b\\
\varepsilon & \varepsilon  
\end{array}
\right].
\end{eqnarray*}
The weighted digraph $G({\cal A})$ is illustrated in figure \ref{fig:A}.
The time evolution of all the position $k$ of the max-plus walk can be written as
\begin{eqnarray*}
\Psi^{n}={\cal{A}}\otimes \Psi^{n-1} = {\cal{A}}^{\otimes n}\otimes \Psi^{0}.
\end{eqnarray*}
\begin{figure}[t!]
\includegraphics[scale=0.5]{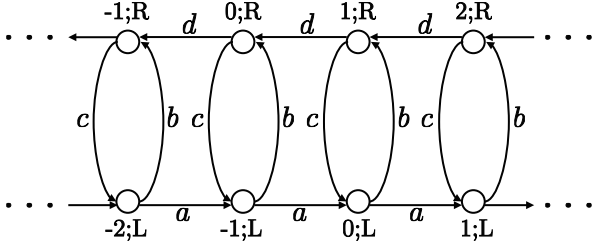}
\caption{The weighted digraph $G({\cal{A}})$.}\label{fig:A}
\end{figure}
Let $k\in \mathbb{Z}$. 
Then the contribution of the state $\psi_{k-i}^{0}$ $(i\in\mathbb{Z})$ for the state $\psi_{k}^{n}$ is $A_{i}^{n}\otimes\psi_{k-i}^{0}$, and so
\begin{eqnarray*}
\psi_{k}^{n} =\bigoplus_{i\in\mathbb{Z}}A_{i}^{n}\otimes\psi_{k-i}^{0}.
\end{eqnarray*}
Therefore, it follows that

\begin{eqnarray*}
\left[
\begin{array}{c}
\vdots\\
\psi_{-1}^{n}\\
\psi_{0}^{n}\\
\psi_{1}^{n}\\
\vdots\\
\end{array}
\right]
=
\left[
\begin{array}{ccccc}
\ddots&\vdots &\vdots &\vdots&\\
\cdots&A_{0}^{n}&A_{-1}^{n}&A_{-2}^{n}&\cdots\\
\cdots&A_{1}^{n}&A_{0}^{n}&A_{-1}^{n}&\cdots\\
\cdots&A_{2}^{n}&A_{1}^{n}&A_{0}^{n}&\cdots\\
&\vdots&\vdots &\vdots &\ddots
\end{array}
\right]
\otimes
\left[
\begin{array}{c}
\vdots\\
\psi_{-1}^{0}\\
\psi_{0}^{0}\\
\psi_{1}^{0}\\
\vdots\\
\end{array}
\right].
\end{eqnarray*}
Hence the $n$-th power of ${\mathcal{A}}$ in (\ref{calA}) is given by 
\begin{eqnarray*}
{\mathcal{A}}^{\otimes n}=
\left[
\begin{array}{ccccc}
\ddots&\vdots &\vdots &\vdots&\\
\cdots&A_{0}^{n}&A_{-1}^{n}&A_{-2}^{n}&\cdots\\
\cdots&A_{1}^{n}&A_{0}^{n}&A_{-1}^{n}&\cdots\\
\cdots&A_{2}^{n}&A_{1}^{n}&A_{0}^{n}&\cdots\\
&\vdots&\vdots &\vdots &\ddots
\end{array}
\right].
\end{eqnarray*}
In the case of the single seed, we discussed in section 3 and section 4, namely,
$
\Psi^{0} = 
\begin{bmatrix}
\cdots&{\cal{E}}&
(\psi_{0}^{0})^{\top}&
{\cal{E}}&
\cdots
\end{bmatrix}^{\top}
$,
$\Psi^{n}$ is given as
\begin{eqnarray*}
\Psi^{n} = 
\left[
\begin{array}{c}
\vdots\\
\psi_{-1}^{n}\\
\psi_{0}^{n}\\
\psi_{1}^{n}\\
\vdots
\end{array}
\right]
=
\left[
\begin{array}{c}
\vdots\\
A_{-1}^{n}\otimes \psi_{0}^{0}\\
A_{0}^{n}\otimes \psi_{0}^{0}\\
A_{1}^{n}\otimes \psi_{0}^{0}\\
\vdots
\end{array}
\right].
\end{eqnarray*}
So we have $\psi_{k}^{n} = A_{k}^{n}\otimes \psi_{0}^{0}$, which is equivalent to (\ref{mpqw_evolution2}).
Then, concerning the spectrum of the infinite matrix $\cal{A}$ in (\ref{calA}), we have the following theorem. 
\begin{theorem}\label{thm:SpectralAnalysis}
For the infinite matrix $\cal{A}$ in (\ref{calA}), 
the spectrum $\sigma(\mathcal{A})$ is
\begin{eqnarray*}
\sigma({\cal{A}}) = \{0\}
\end{eqnarray*}
and the eigenvector $v({\cal{A}})$ corresponding to the eigenvalue 0 is 
\begin{eqnarray*}
v({\cal{A}})=
\kappa\otimes 
\left[
\begin{array}{c}
\vdots \\
X_{-1}\\
X_0\\
X_{1}\\
\vdots
\end{array}
\right], \quad 
X_k=
\left[
\begin{array}{c}
-ak\\
(-k+1)a -b\\
\end{array}
\right],
\end{eqnarray*}
where $\kappa \in \mathbb{R}_{\max}$ is constant. 
\end{theorem}
\begin{proof}

In figure \ref{fig:A}, the maximum average weight of circuits is 0.
From proposition \ref{prop:mp_eig}, the spectrum of ${\cal{A}}$ is obviously \{0\}.
Next we consider the eigenvectors of ${\cal{A}}$.
Since the eigenvalue of ${\cal{A}}$ is 0,
$(-\lambda\otimes {\cal{A}})^{*}$ is equivalent to ${\cal{A}}^{*}$, which is given as
\begin{eqnarray}
{\cal{A}}^{*}=I \oplus {\cal{A}} \oplus  {\cal{A}}^{\otimes 2} \oplus \cdots \oplus {\cal{A}}^{\otimes n}\oplus \cdots,\label{A_star}
\end{eqnarray}
where $I$ is the infinite max-plus identity matrix.
Since all the vertices of $G({\cal{A}})$ are contained in a circuit with the average weight 0,
from proposition \ref{prop:mp_evec},
any columns of ${\cal{A}}^{*}$ is an eigenvector. 
Any neighboring two columns of the right hand side of (\ref{A_star}) is given as
\begin{eqnarray}
\bigoplus_{n}
\left[
\begin{array}{c}
\vdots \\
A_{-1}^{n}\\
A_{0}^{n}\\
A_{1}^{n}\\
\vdots
\end{array}
\right]
=
\left[
\begin{array}{c}
\vdots \\
A_{-1}\\
A_{0}\\
A_{1}\\
\vdots
\end{array}
\right],\quad
A_{k} =
\left[ 
\begin{array}{cc}
-ka&(-k-1)a+b\\
(-k+1)a-b&-ka\\
\end{array}
\right], \quad k\in\mathbb{Z},
\label{eq:AAA}
\end{eqnarray}
since entries of $A_{k}^{n}$ in (\ref{A_k^n}) is not depend on $n$.
Then, the right hand side of (\ref{eq:AAA}) can be written down as
\begin{eqnarray*}
\left[
\begin{array}{cc}
\vdots&\vdots \\ 
a&b\\
2a-b&a\\ 
0&-a+b\\
a-b&0\\ 
-a&-2a+b\\ 
-b&-a\\ 
\vdots&\vdots
\end{array}
\right]
=
\left[
\begin{array}{cc}
\vdots &\vdots \\
X_{-1}&(-a+b)\otimes X_{-1}\\
X_{0}&(-a+b)\otimes X_{0}\\
X_{1}&(-a+b)\otimes X_{1}\\
\vdots&\vdots 
\end{array}
\right].
\end{eqnarray*}
Therefore, each column of the above right hand side is the eigenvector of ${\cal{A}}$, without loss of generality. 
\end{proof}
The eigenvector $v({\cal{A}})$ is set to be the initial state of the max-plus walk, namely, 
\begin{eqnarray*}
\Psi^{0}=
\left[
\begin{array}{c}
\vdots\\
X_{-1}\\
X_{0}\\
X_{1}\\
\vdots
\end{array}
\right].
\end{eqnarray*}
Then it holds that for any $n=0,1,\dots$,
\begin{eqnarray*}
\Psi^{n+1}={\cal{A}}\otimes\Psi^{n} = \Psi^{n},
\end{eqnarray*}
which gives the stationary state of  the max-plus walk independent of $n$.

\section{Concluding remarks and discussions}

\begin{table}
\caption{Comparison between the quantum walk and the max-plus walk} \label{table_comp}
\begin{tabular}{ccc}
&Quantum walk &Max-plus walk\\[3pt] 
Operations&$+, \times $&$\oplus, \otimes$\\[3pt] \hline

Time evolution&$\varphi_{k}^{n} = {\bf P}\varphi_{k+1}^{n-1} + {\bf Q}\varphi_{k-1}^{n-1}$&$\psi_{k}^{n} = P\otimes\psi_{k+1}^{n-1} \oplus Q\otimes\psi_{k-1}^{n-1}$\\[3pt] \hline
\begin{tabular}{c}
Conserved\\
quantities\\
w.r.t. $n$
\end{tabular}
&
\begin{tabular}{c}
$\displaystyle\sum_{k\in\mathbb{Z}}\|\varphi_{k}^{n}\|_{\mathbb{C}^{2}}^{2}$
\\
$=\displaystyle\frac{1}{2}\sum_{k\in\mathbb{Z}}\|{\bf A}_{k}^{n}\|_{F}^{2}=1$
\end{tabular}
&$\displaystyle\sum_{k\in\mathbb{Z}}\lambda(A_{k}^{n})=0$
\\[3pt]  \hline
Quantum coin &
\begin{tabular}{c}
Unitary $\Leftrightarrow$ \\
$\left\{\begin{array}{l}
|{\rm a}|^{2} + |{\rm b}|^{2}= |{\rm c}|^{2}+|{\rm d}|^{2}=1, \\
{\rm a}\bar{{\rm c}}+{\rm b} \bar{{\rm d}}=0
\end{array}
\right.
$
\end{tabular}
&
\begin{tabular}{c}
$a+d = b+c=0$\\
(theorem 4.1)
\end{tabular}
\\[3pt] \hline
%
%
Spectrum &
\begin{tabular}{c}
${\rm Spec}({\cal U})$\\ 
$=\{\Delta^{\prime 1/2}e^{i\theta}: \cos\theta\in[-|{\rm a}|,|{\rm a}|]\}$
\end{tabular}
&
\begin{tabular}{c}
$\sigma({\cal A}) = \{0\}$\\
(theorem 5.1)
\end{tabular}
\\[3pt]  \hline
\begin{tabular}{c}
Eigenvector\\
or generalized \\
eigenfunction \\ 
\end{tabular}
&
\begin{tabular}{c}
bounded \cite{Konno2014,Komatsu2017,Morioka},\\
quadratical \cite{Konno2015},\\
exponential \cite{Komatsu2019} 
\end{tabular}
&$v({\cal A}) \propto$ linear \\ 
\end{tabular}
\end{table}

Analysis on probability cellular automata using quantum operators and descriptions were considered in \cite{SudLl1,SudLl2}. 
The Domany-Kinzel model is a famous probability cellular automaton model having expression of 
oriented site-bond percolation with the signed measure~\cite{KKST,KonnoDuality}. 
An expression with tensorproduct operator, which is usually used in quantum information theory, appears to express the Domany-Kinzel model
and reveals interesting duality relations by using this expression well \cite{KKST}. 
These fascinate trials to approach probabilistic problems through quantum operations  
provide expectation to us solving great open problems e.g., mathematically solving the phase diagram of the Domany-Kinzel model 
using quantum phenomena. 
Then it is reasonable to target quantum walks for discussing relations between cellular automata and quantum phenomena since 
quantum walks are expected to implement several kinds of quantum phenomenon on quantum devices.  
Various kinds of cellular automata are obtained by taking ultradiscretization to difference equations. 
It is possible to take its inverse operation from cellular automata to difference equations, 
for example, a difference equation induced by Domany-Kinzel model can be seen in \cite{KonnoKunimatsuMa}.
Then as a first trial to tackle this challenging problem, in this paper, we proposed a new model which is analogous to quantum walks over the max-plus algebra. 
We named it the max-plus walk.

We summarized our results in table \ref{table_comp}.
We obtained an explicit expression of weight of paths, namely the state decision matrices (\ref{eq:SDM}) which is corresponding to (\ref{eq:QW_amplitude}). 
The state decision matrices (\ref{eq:SDM})
can be regarded as an ultradiscrete limit of (\ref{eq:QW_amplitude}). 
In the conventional quantum walk, the summation of the $\ell^2$-norm of the state vector and also the Frobenius norm of the state decision matrices are conserved quantities.
It is remarkable here that the Frobenius norm can be described by singular values.
As an analogue of it, we derive a conserved quantity of the max-plus walk by the summation over all the positions of eigenvalues of the state decision matrices. 
In the quantum walk, in order to conserve the Frobenius norm, the quantum coin ${\bf P} + {\bf Q}$ needs to be unitary.
On the other hand, in the max-plus walk,  
the condition $a+d=b+c=0$ is the necessary and sufficient condition for the conservation and the conserved value to be 0 (theorem 4.1). 

While in the quantum walk, the spectrum continuously lies on the unit circle whose real part is $[-|{\rm a}|, |{\rm a}|]$ if $\Delta^{\prime 1/2} = 1$, obtained by functional analysis approach, 
in the max-plus walk,
the spectrum is $\{0\}$ which is obtained by a graph theoretical approach.
For the quantum walk on $\mathbb{Z}$,
a classification of the generalized eigenfunctons' shape
are investigated  with respect to the absolute values, from the view point of spectral analysis e.g.,
boundedness \cite{Konno2014,Komatsu2017,Morioka}, 
nonboundedness; polynomially increasing \cite{Konno2015} and exponentially increasing \cite{Komatsu2019}.
On the other hand, 
in the max-plus walk,
since both $2k$ and $2k+1$ components of the eigenvector are proportional to $k$,
the absolute values of the entries of the eigenvector are linearly increasing. 

In a separate paper, we will discuss an analogue of the limit theorem (\ref{eq:Konno}) of our max-plus walk. 
To obtain relationships of our max-plus walk to 
quantum cellular automata also remains open.


\par
\
\par\noindent
\noindent
{\bf Acknowledgments.}
\par
This work was partially supported by Grants-in-Aid for Scientific Research (C) No. 19K03616 and  No. 19K03624 
 of
the Japan Society for the Promotion of Science and 
Research Origin for Dressed Photon.
\par



\end{document}